\documentclass[journal]{IEEEtran}
\usepackage{amsmath,amsfonts,amssymb,mathrsfs}

\usepackage{algorithm}
\usepackage{algpseudocode}

\usepackage{cite}

\usepackage[caption=false,font=footnotesize]{subfig}
\usepackage{graphicx}

\usepackage{tikz, pgfplots}
\usetikzlibrary{shapes, arrows, arrows.meta, fit, positioning, calc, 3d}
\usepgfplotslibrary{fillbetween}
\pgfplotsset{compat=1.18} 
\tikzset{>=latex}

\newtheorem{theorem}{Theorem}

\newtheorem{remark}{Remark}
\newtheorem{definition}{Definition}

\newtheorem{corollary}{Corollary}

\allowdisplaybreaks
\usepackage[acronym]{glossaries}
\newcommand{\newac}{\newacronym}
\newcommand{\ac}{\gls}

\newac{rs}{RS}{randomized stationary}
\newac{iid}{\emph{i.i.d.}}{independent and identically distributed}
\newac{zoh}{ZOH}{zero-order hold}
\newac{aoi}{AoI}{age of information}
\newac{iot}{IoT}{Internet of Things}
\newac{iaws}{IAWS}{independence-assumed weighted-sum}
\newac{cra}{CRA}{confidential reconstruction accuracy}
\newac{los}{LOS}{line of sight}
\newac{nlos}{NLOS}{non-line of sight}

\begin{document}
\bstctlcite{IEEEtran:no_dash_repeated_names}

\title{Joint Accuracy and Confidentiality in Semantic-Aware Secure Remote Reconstruction}
\author{Bowen~Li,~\IEEEmembership{Member,~IEEE},~and~Nikolaos~Pappas,~\IEEEmembership{Senior~Member,~IEEE}
\thanks{The authors are with the Department of Computer and Information Science, Link\"oping University, Link\"oping 58183, Sweden (e-mail: bowen.li@liu.se; nikolaos.pappas@liu.se). This work was supported by ELLIIT, and the European Union (6G-LEADER, 101192080, and MAGIC-6G, 101292933).}
}

\maketitle
\begin{abstract}
In this paper, we consider remote reconstruction over wireless networks when simultaneous accuracy at the legitimate receiver and confidentiality against eavesdropping are required. These two objectives are often treated separately, even though they arise from the same update process and are marginals of a joint reconstruction event. This paper introduces confidential reconstruction accuracy (CRA), a metric to capture the joint event in which the legitimate receiver reconstructs correctly while the eavesdropper fails. Under randomized stationary (RS) policies, we develop a three-dimensional stationary analysis and derive closed-form expressions for the long-term average CRA and the optimal transmission probability. The results show that conventional marginal analysis may misidentify the optimal RS policy and misestimate the achievable simultaneous accuracy-confidentiality performance. Numerical results reveal nontrivial behaviors: more frequent transmissions or better legitimate channels do not necessarily improve joint accurate and confidential reconstruction, and when the eavesdropping channel is strong, improving the legitimate channel alone may be insufficient. Finally, the framework induces the spatial safety boundary in a geofencing setting for secure remote reconstruction.
\end{abstract}


\section{Introduction}
Remote reconstruction over wireless networks is fundamental to cyber-physical systems and \ac{iot} applications~\cite{KouPap:J21,GilYemChoNed:J25}. Owing to the broadcast nature of wireless propagation, transmitted updates are inherently vulnerable to eavesdropping~\cite{LiuXiaLiLia:J12}. In adversarial environments, each update jointly shapes both the legitimate receiver’s and the eavesdropper’s reconstruction~\cite{GopLaiEl:J08,TsiGatPap:J20}, so the value of the update cannot be assessed from the legitimate side alone. Hence, the fundamental objective is not accuracy alone, nor confidentiality alone, but the joint event that the legitimate receiver reconstructs the source correctly while the eavesdropper fails.

Classical secure communication prioritizes bit-level protection against eavesdropping, with representative metrics such as secrecy capacity and equivocation characterizing achievable confidential rates and the eavesdropper’s uncertainty about the message~\cite{OggHas:J11,LeuHel:J78}. However, these metrics do not account for the task-specific importance of information. In remote reconstruction, this limitation is critical~\cite{SunUysEliYat:J17}; transmitting redundant or predictable updates provides negligible accuracy gains while unnecessarily increasing information leakage.
Recent semantic-aware research addresses this issue by optimizing updates based on the value of information~\cite{LuoDelSalPap:A25,LuoLiPap:J26}, yet these frameworks typically focus on the legitimate side, leaving confidentiality unaddressed at the semantic level.

Some existing frameworks consider both accuracy and confidentiality, for example, through threshold-based constraints, Pareto formulations, or weighted combinations of separate metrics~\cite{ZhaXuLanChe:J24,LiLiLvJu:J25,WanGuoWanBai:J25}. However, they characterize the legitimate and adversarial sides separately and combine them afterward. In secure remote reconstruction, such formulations do not directly capture the event of interest, because the same transmission policy jointly governs legitimate reconstruction and adversarial exposure through a common update process. Accordingly, accuracy and confidentiality are not simply two separate objectives, but different marginals of a joint reconstruction process. The relevant target is therefore the joint event that the legitimate receiver reconstructs correctly while the eavesdropper fails.

This work addresses two fundamental questions: \emph{(i) how to characterize joint accuracy and confidentiality in secure remote reconstruction}, and \emph{(ii) how this joint characterization affects transmission policy design}. To this end, we investigate semantically aware confidential remote reconstruction under \ac{rs} policies, making three main contributions. First, we introduce the \ac{cra} metric for joint accurate-and-confidential reconstruction and derive closed-form expressions for its long-term average and the optimal transmission probability via three-dimensional Markov chain stationary analysis. Second, we show that conventional marginal analysis fails to capture the intrinsic coupling between legitimate reconstruction and adversarial exposure, which leads to suboptimal \ac{rs} policies and misestimates simultaneous accuracy-confidentiality performance; the joint characterization further reveals nontrivial structural behaviors, including that more frequent transmissions or better legitimate channels do not necessarily improve simultaneous accurate-and-confidential reconstruction, and that improving the legitimate channel alone may be insufficient when the eavesdropping channel is strong. Finally, we apply the proposed framework to a geofencing scenario, where the analytical results induce the spatial safety boundary that restricts eavesdropper proximity to ensure secure remote reconstruction.

\begin{figure}
    \centering
    \includegraphics[width=1\linewidth]{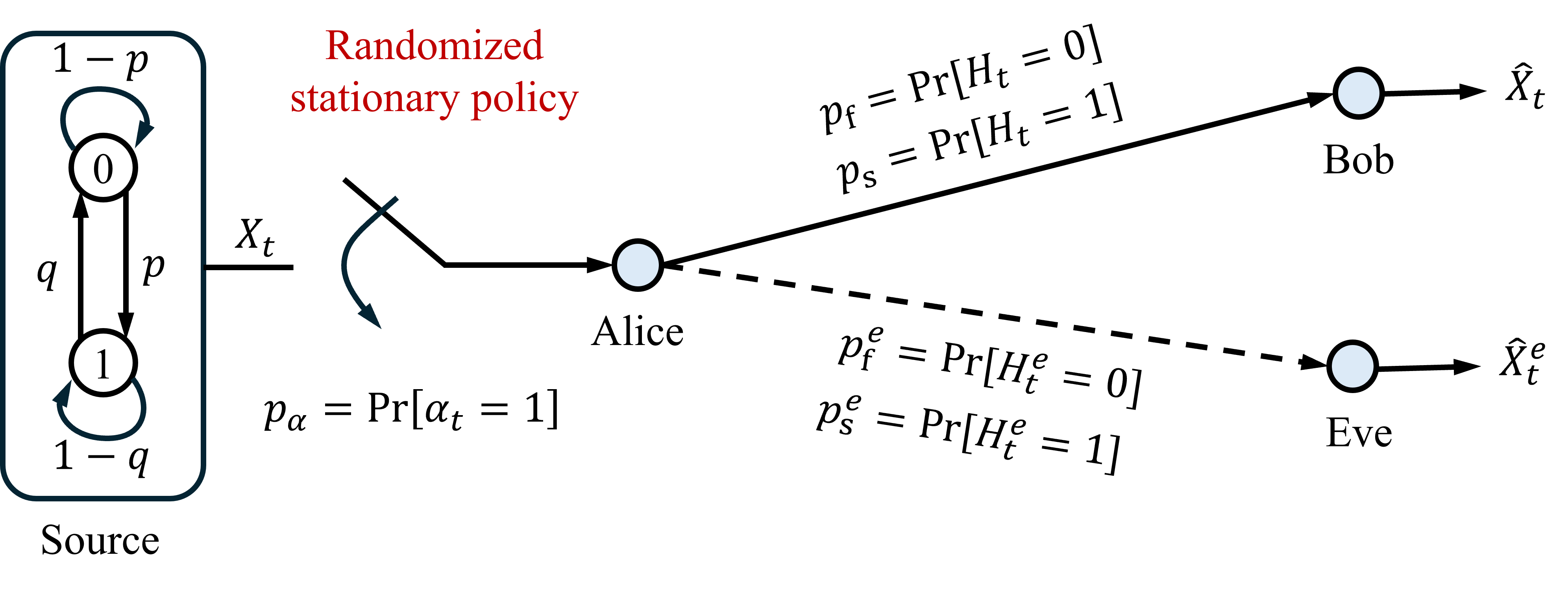}
    \caption{Remote reconstruction system over a wiretap channel with a transmitter (Alice), a legitimate receiver (Bob), and an eavesdropper (Eve).}
    \label{fig:cra_model}
\end{figure}

\section{System Model}\label{sec:system_model}
We consider a remote reconstruction system over a wiretap channel as in Fig.~\ref{fig:cra_model}, comprising a source, a transmitter (Alice), a legitimate receiver (Bob), and an eavesdropper (Eve).
\subsubsection{Source}
The source process is modeled as a time-homogeneous binary Markov chain $\{X_t\}_{t\geq 1}$, where $X_t\in\{0,1\}$ denotes the system state at slot $t\in\{1,2,\cdots\}$.
The state transition matrix is given by
\begin{equation*}
Q=\begin{bmatrix}
1-p & p\\
q & 1-q
\end{bmatrix},
\end{equation*}
where $p=\Pr[X_{t+1}=1\mid X_t=0]$ and $q=\Pr[X_{t+1}=0\mid X_t=1]$.

We assume $0<p,q<1$, so that the chain is irreducible and aperiodic (ergodic), and hence admits a unique stationary distribution $\boldsymbol{v}=[v_0,v_1]^T$ with $v_0=\frac{q}{p+q}$, $v_1=\frac{p}{p+q}$. The parameters $p$ and $q$ further characterize the temporal correlation of the source: $p+q<1$ corresponds to a persistent (positively correlated) process, $p+q=1$ corresponds to a memoryless (\ac{iid}) process, and $p+q>1$ corresponds to an alternating (negatively correlated) process~\cite{SeaWis:J23}.

\subsubsection{Transmission policy (Alice)}
At each slot $t$, Alice makes a transmission decision $\alpha_t\in\{0,1\}$, where $\alpha_t=1$ denotes sampling and transmitting a perfect observation of the source state $X_t$, and $\alpha_t=0$ denotes idleness. We focus on \ac{rs} policies, under which $\{\alpha_t\}$ is an \ac{iid} Bernoulli process independent of the source evolution, with
\begin{equation*}
\Pr[\alpha_t=1]=p_\alpha,\qquad \Pr[\alpha_t=0]=1-p_\alpha,
\end{equation*}
for all $t\geq 1$, where $p_\alpha\in(0,1]$ is the transmission probability.

\subsubsection{Wiretap Channel}
We consider a wiretap channel from Alice to Bob and Eve, where each transmission may be successfully received or erased. The main channel to Bob and the eavesdropping channel to Eve are modeled by two independent \ac{iid} Bernoulli processes, denoted by $\{H_t\}_{t\geq 1}$ and $\{H_t^e\}_{t\geq 1}$, respectively, with
\begin{equation*}
    \Pr[H_t=1]=p_s,\qquad \Pr[H_t^e=1]=p_s^e,
\end{equation*}
for all $t\geq 1$, where $p_s,\,p_s^e\in (0,1)$. $H_t=1$ means that the update transmitted at slot $t$ is successfully received by Bob, and $H_t^e=1$ means that the same update is successfully intercepted by Eve. The corresponding erasure probabilities are $1-p_s$ and $1-p_s^e$, respectively.

\subsubsection{Receiver Observation and Reconstruction (Bob and Eve)}
The observations at Bob and Eve are jointly determined by Alice's transmission decision and the corresponding channel realizations. Specifically, Bob receives
\begin{equation}
    Y_t=
    \begin{cases}
    X_t, & \text{if } \alpha_t=1 \text{ and } H_t=1,\\
    \varepsilon, & \text{otherwise},
    \end{cases}
    \label{eq:def_receive_model}
\end{equation}
where $\varepsilon$ denotes an erasure symbol. Similarly, Eve observes
\begin{equation}
    Y_t^e=
    \begin{cases}
    X_t, & \text{if } \alpha_t=1 \text{ and } H_t^e=1,\\
    \varepsilon, & \text{otherwise}.
    \end{cases}
    \label{eq:def_receive_model_eve}
\end{equation}

We assume that the receivers have no information about the source dynamics. Accordingly, both receivers employ a model-free estimator that uses the most recently received update. Let $\Theta_t$ and $\Theta_t^e$ denote the \ac{aoi} at Bob and Eve, respectively, \emph{i.e.}
\begin{equation}
    \Theta_t=t-\max\{\tau\le t: Y_\tau\neq \varepsilon\},
    \label{eq:def_aoi}
\end{equation}
and $\Theta_t^e$ is defined analogously using $Y_\tau^e$. The corresponding estimates are
\begin{equation}
    \hat X_t=X_{t-\Theta_t}, \qquad \hat X_t^e=X_{t-\Theta_t^e}.
    \label{eq:def_estimator}
\end{equation}

\subsubsection{Confidential Reconstruction Accuracy}
In secure remote reconstruction, the relevant target is the joint event that the legitimate receiver is correct while the eavesdropper is not, rather than a post hoc combination of separate accuracy and confidentiality marginals, \emph{e.g.}, \cite{SalKouPap:J24}. To characterize this event directly, we introduce \ac{cra}.
\begin{definition}
The instantaneous \ac{cra} at time $t$ is defined as
\begin{equation}
    A_{\text{CRA}}[t] := \mathbb{I}\left\{ \hat{X}_t = X_t,\ \hat{X}_t^e \neq X_t \right\},
    \label{eq:A_sec_def}
\end{equation}
where $\mathbb{I}\{\cdot\}$ denotes the indicator function. $A_{\text{CRA}}[t]=1$ if and only if the legitimate receiver correctly reconstructs the source state while the eavesdropper fails.
\end{definition}

The long-term average \ac{cra} is defined as
\begin{equation}
    \bar{A}_{\text{CRA}} := \liminf_{T\to\infty} \frac{1}{T} \sum_{t=1}^{T} \mathbb{E}\left[A_{\text{CRA}}[t]\right].
    \label{eq:avg_cra}
\end{equation}

\section{Analytical Results}
This section characterizes the joint source-reconstruction dynamics. By deriving the stationary distribution of the joint state, we obtain a closed-form expression for the long-term average \ac{cra}. This expression further enables the derivation of the optimal \ac{rs} policy. 

\subsection{System Dynamics}\label{subsec:system_dynamics}
The joint evolution of the source and estimators is captured by the composite state
$\mathbf{S}_t := (X_t, \hat{X}_t, \hat{X}^e_t) \in \{0,1\}^3.$ Given that the source $\{X_t\}$ is Markovian and the transmission decisions $\{\alpha_t\}$ and channel processes $\{H_t,H_t^e\}$ are mutually independent \ac{iid} sequences, $\{\mathbf{S}_t\}_{t \geq 1}$ constitutes a finite-state time-homogeneous ergodic Markov chain, as proved below. 

Since $\{X_t\}$ is Markovian, while the transmission decisions $\{\alpha_t\}$ and the channel processes $\{H_t\}$ and $\{H_t^e\}$ are \ac{iid}, mutually independent.
The updates of $\hat X_t$ and $\hat X_t^e$ depend only on the current state $(X_t,\hat X_t,\hat X_t^e)$ and the current realization of $(\alpha_t,H_t,H_t^e)$. Hence,
\begin{equation*}
\Pr(\mathbf{S}_{t+1}\mid \mathbf{S}_1,\ldots,\mathbf{S}_t)
=
\Pr(\mathbf{S}_{t+1}\mid \mathbf{S}_t).
\end{equation*}

For $\mathbf{s}=(x,a,b)$ and $\mathbf{s}'=(x',a',b')$, define the one-step transition probability as
\begin{align*}
    \Pr (\mathbf{s}, \mathbf{s}') 
    & := \Pr (\mathbf{S}_{t+1}=\mathbf{s}' \mid \mathbf{S}_{t} = \mathbf{s}) \\
    & = \Pr ( X_{t+1}=x',\hat{X}_{t+1}=a', \hat{X}^e_{t+1}=b' \mid \\
    & \qquad\quad X_{t}=x, \hat{X}_t=a, \hat{X}^e_t=b ).
\end{align*}
By the chain rule, $\Pr (\mathbf{s}, \mathbf{s}')$ becomes
\begin{align}
    &\Pr ( X_{t+1}=x' \mid X_{t}=x, \hat{X}_t=a, \hat{X}^e_t=b ) \cdot\Pr (\hat{X}_{t+1}=a', \nonumber\\
    & \qquad\hat{X}^e_{t+1}=b' \mid X_{t}=x, \hat{X}_t=a, \hat{X}^e_t=b, X_{t+1}=x' ),
    \label{eq:chain_rule_for_tp}
\end{align}
Since the source is independent of transmission decisions and channel processes, the first term in \eqref{eq:chain_rule_for_tp} becomes $\Pr ( X_{t+1}=x' \mid X_{t}=x) $. 

The second term is determined by the joint reception outcome. Define
\begin{align*}
    \lambda_{11} &:= p_{\alpha} p_s p_s^e, & (\text{Both reconstruct}) \\
    \lambda_{10} &:= p_{\alpha} p_s (1-p_s^e), & (\text{Only Bob reconstructs}) \\
    \lambda_{01} &:= p_{\alpha} (1-p_s) p_s^e, & (\text{Only Eve reconstructs}) \\
    \lambda_{00} &:= p_{\alpha} (1-p_s)(1-p_s^e) + \bar{p}_{\alpha}. & (\text{Neither reconstructs})
\end{align*}
Then the second term of \eqref{eq:chain_rule_for_tp}, denoted as $\mathcal{P}_\text{s}$, becomes 
\begin{align*}
\mathcal{P}_\text{s} &= \lambda_{11}\mathbb{I}(a'=x',b'=x')+\lambda_{10}\mathbb{I}(a'=x',b'=b) \\
& \quad +\lambda_{01}\mathbb{I}(a'=a,b'=x')+ \lambda_{00}\mathbb{I}(a'=a,b'=b).
\end{align*}

Combining the above, the transition kernel is given by 
\begin{equation}
    \Pr (\mathbf{s}, \mathbf{s}')  = \Pr ( X_{t+1}=x' \mid X_{t}=x)  \cdot\mathcal{P}_\text{s}.
\end{equation}
Since the source process is time-homogeneous and $\mathcal{P}_s$ is independent of $t$, the Markov chain $\{\mathbf{S}_t\}$ is time-homogeneous.

Since $p_\alpha\in(0,1]$, $p_s\in(0,1)$, and $p_s^e\in(0,1)$, all $\lambda$-parameters are strictly positive. Together with the irreducibility of $Q$, this ensures that all states of $\{\mathbf{S}_t\}$ communicate, so the chain is irreducible. Furthermore, because $\lambda_{00}>0$ and $Q(x,x)>0$ for some source state $x$, the chain has a self-loop and is aperiodic. Therefore, the finite-state chain $\{\mathbf{S}_t\}$ is ergodic. As a result, it admits a unique stationary distribution, and the long-term average exists.

\subsection{Stationary Distribution}
This subsection derives the joint stationary distribution of $\mathbf{S}_t$. Specifically, for any $(x,a,b)\in\{0,1\}^3$, define
\begin{equation*}\pi(x, a, b) := \lim_{t \to \infty} \Pr(X_t=x, \hat{X}_t=a, \hat{X}^e_t=b).\end{equation*}
The following theorem gives its closed-form expression.

\begin{theorem}\label{thm:stationary}
The joint stationary distribution $\pi(x, a, b)$ is
\begin{equation}
    \pi(x, a, b) = T_1(x,a,b) + T_2(x,a,b) + T_3(x,a,b),
\end{equation}
where
\begin{align}
    T_1 (x, a, b) &= \delta(a,b) v_a \lambda_{11} \left[ (I - \lambda_{00}Q)^{-1} \right]_{a,x}, \\
    T_2 (x, a, b) &= v_b \lambda_{10} P_B \left[ Q \left(I - (1-P_B)Q\right)^{-1} \right]_{b,a} \nonumber \\
        &\quad \times \left[ (I - \lambda_{00}Q)^{-1} \right]_{a,x}, \\
    T_3 (x, a, b) &= v_a \lambda_{01} P_A \left[ Q \left(I - (1-P_A)Q\right)^{-1} \right]_{a,b} \nonumber \\
        &\quad \times \left[ (I - \lambda_{00}Q)^{-1} \right]_{b,x}.
\end{align}
Here, $[M]_{i,j}$ denotes the entry of matrix $M$ with row index $i$ and and column index $j$, $\delta(a,b)$ is the Kronecker delta, and $P_A := \lambda_{11} + \lambda_{10}$ and $P_B := \lambda_{11} + \lambda_{01}$ denote the marginal success probabilities at Bob and Eve, respectively.
\end{theorem}

\begin{IEEEproof}
    See Appendix~A.
\end{IEEEproof}
It follows from Theorem~\ref{thm:stationary} that $\pi(x,a,b)$ is rational in the policy parameter $p_\alpha$, which makes the joint-event metric analytically tractable under randomized stationary policies. Indeed, $\lambda_{00},\lambda_{10},\lambda_{01},\lambda_{11}$, and hence $P_A$ and $P_B$, are affine in $p_\alpha$. Moreover, for $c\in\{\lambda_{00},1-P_A,1-P_B\}$, the resolvent matrix $(I-cQ)^{-1}$ is rational in $c$, so these resolvent matrices are rational in $p_\alpha$ since each $c$ is affine in $p_\alpha$. As a result, $T_1(x,a,b)$, $T_2(x,a,b)$, $T_3(x,a,b)$, and hence $\pi(x,a,b)$ are rational in $p_\alpha$. This rational structure will be exploited to derive closed-form expressions for the long-term average \ac{cra} and the corresponding optimal \ac{rs} policy.

\subsection{Average \ac{cra}}
By the definition of $A_{\text{CRA}}[t]$ in \eqref{eq:A_sec_def}, the event of successful confidential reconstruction occurs if and only if $(X_t, \hat{X}_t, \hat{X}^e_t)=(0,0,1)$ or $(1,1,0)$. Since the process $\{\mathbf{S}_t\}$ is ergodic as analyzed in Section~\ref{subsec:system_dynamics}, we have
\begin{equation*}
    \bar{A}_{\text{CRA}} = \pi (0,0,1) +\pi (1,1,0).
\end{equation*}
By Theorem~\ref{thm:stationary}, each stationary probability $\pi (x,a,b)$ is rational in $p_{\alpha}$. Therefore, $\bar{A}_{\text{CRA}}(p_{\alpha})$ is also rational in $p_{\alpha}$, which yields the following corollary.
\begin{corollary} \label{cor:average_cra}
The long-term average \ac{cra} is given by
\begin{equation*}
\bar{A}_{\text{CRA}}(p_{\alpha})
= \frac{A\,p_{\alpha}+B} {C\,p_{\alpha}^{2}+D\,p_{\alpha}+E},\label{eq:opt_p}
\end{equation*}
where
\begin{align*}
A &= p q \left[ p_s^2 (1-p_s^e)(p+q-2) + (p_s^e)^2 (1-p_s) (p+q) \right],\\
B &= p q (p+q) (2 p_s p_s^e - p_s - p_s^e),\\
C &= (p_s p_s^e - p_s - p_s^e) p_s p_s^e (p+q-1)^2 (p+q), \\
D &= (p_s p_s^e - p_s - p_s^e) ( p_s + p_s^e) (1- p-q)(p+q)^2, \\
E &= (p_s p_s^e - p_s - p_s^e) (p+q)^3.
\end{align*}
\end{corollary}

\begin{remark}
Corollary~\ref{cor:average_cra} is stated for the general case $p_s\neq p_s^e$. The symmetric case $p_s = p_s^e$ is degenerate, since the numerator and denominator share a common factor. In that case, the expression should be simplified before substitution to avoid an artificial zero denominator.
\end{remark}

\subsection{Optimal \ac{rs} Policy}
We consider the problem of maximizing the long-term average \ac{cra} over the transmission probability $p_\alpha$
\begin{equation}
    \mathscr{P}1: \quad \max_{ 0 < p_{\alpha} \le 1} \quad \bar{A}_{\text{CRA}}.\label{eq:def_P1}
\end{equation}

Corollary~\ref{cor:average_cra} shows that $\bar{A}_{\text{CRA}}(p_\alpha)$ admits a closed-form rational form, which makes $\mathscr{P}1$ analytically tractable. The resulting optimal transmission probability is given next.

\begin{theorem} \label{thm:opt_P1}
The optimal solution to Problem $\mathscr{P}1$ is 
\begin{equation*}
p_\alpha^*= \left[\dfrac{-BC+\sqrt{\Delta}}{AC}\right]_0^1,
\end{equation*}
where $\left[x\right]_0^1:=\min\{1,\max\{x,0\}\}$ denotes the Euclidean projection of $x$ onto the feasible interval $[0,1]$ and $\Delta:=B^2C^2+AC(AE-BD)$, where $\Delta$ is proven to be nonnegative.
\end{theorem}
\begin{IEEEproof}
    See Appendix~B.
\end{IEEEproof}

\begin{remark}
The result of Theorem~\ref{thm:opt_P1} extends directly to any interval constraint \(p_\alpha\in[p_\ell,p_u]\) with \(0< p_\ell\le p_u\le 1\). In that case, the optimal solution is obtained by projecting the same closed-form candidate onto \([p_\ell,p_u]\), namely
\begin{equation*}
p_\alpha^*
=
\left[\frac{-BC+\sqrt{\Delta}}{AC}\right]_{p_\ell}^{p_u},
\end{equation*}
with \(
[x]_{p_\ell}^{p_u}:=\min\{p_u,\max\{x,p_\ell\}\}.
\)
\end{remark}

\begin{remark}\label{rem:extrem_analysis_opt}
    The expression in Theorem~\ref{thm:opt_P1} corresponds to the general case. In the special case \(p_s=p_s^e\), \(\bar{A}_{\text{CRA}}(p_\alpha)\) is monotonic over $p_\alpha\in(0,1]$: it is increasing when \(p+q>1\) and decreasing when \(p+q<1\), and constant when \(p+q=1\). Accordingly, $p_\alpha^*=1$ if $p+q>1$, $p_\alpha^*\to0^+$ if $p+q<1$, and $p_\alpha^*$ can be any value in $(0,1]$ if $p+q=1$.
 
    Thus, in the symmetric-channel case, \emph{the optimal \ac{rs} policy is determined entirely by the source temporal correlation}: persistent sources favor silence, alternating sources favor full transmission, and the \ac{iid} case is indifferent to the transmission probability.
\end{remark}

\section{Numerical and Simulation Results}
In this section, we first validate the analytical results using Monte Carlo simulations with horizon $T=50,000$, averaged over $400$ independent realizations. We then show how the proposed joint-event metric differs from conventional separate-marginal formulations, and reveal structural behaviors that are invisible under marginal analysis. Finally, we demonstrate a system-level implication of the framework through a geofencing example, where the \ac{cra} metric is used to construct the corresponding secure region and boundary.
\begin{figure}
    \centering
    \includegraphics[width=1\linewidth]{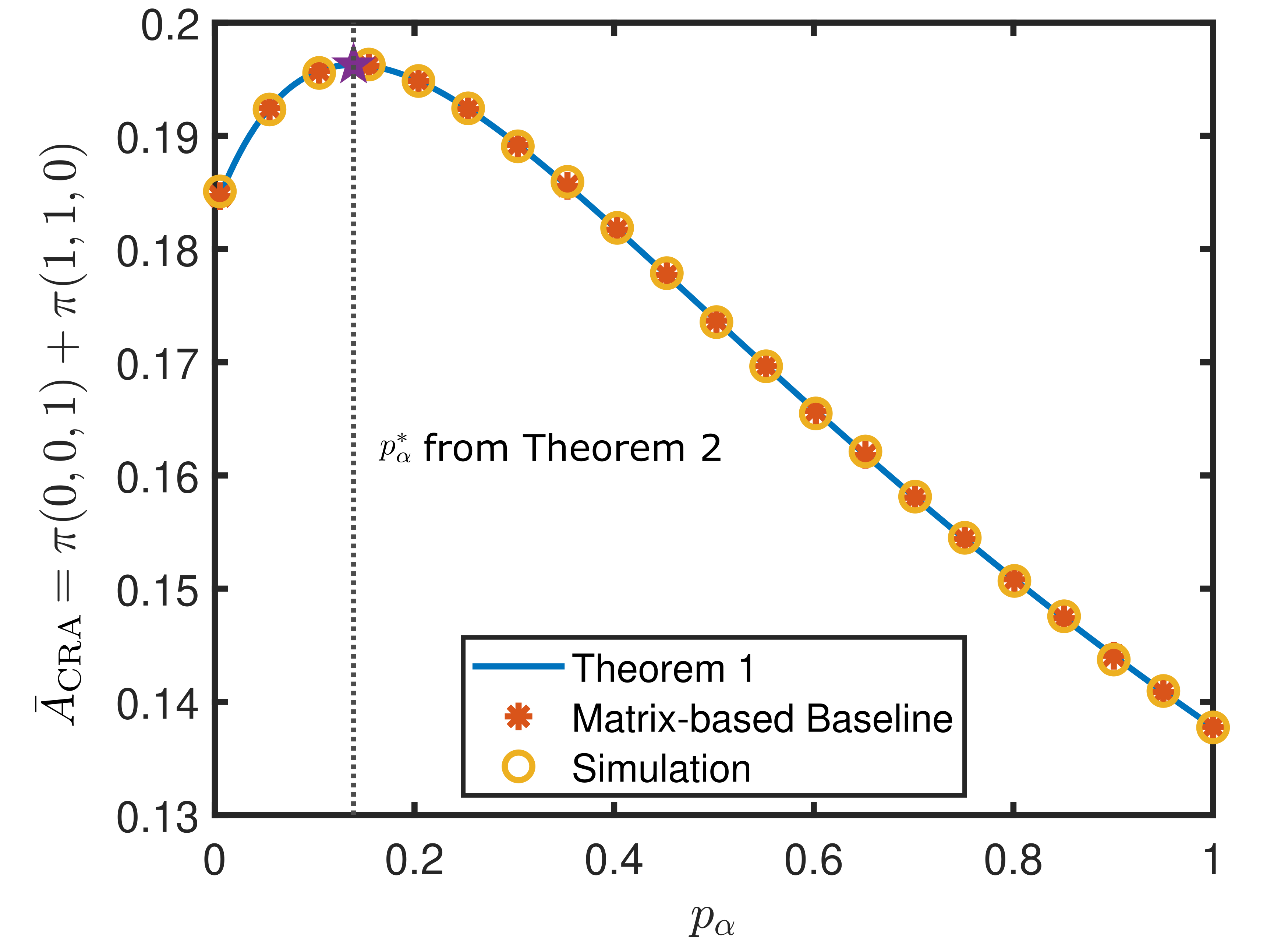}
    \caption{Average simultaneous confidentiality and accuracy $\bar{A}_\text{CRA}$ versus sampling and transmission probability $p_\alpha$.}
    \label{fig:Monto_carlo}
\end{figure}

For comparison, we consider a marginal-analysis baseline, which approximates the joint confidentiality-accuracy objective by separately evaluating the marginal accuracy and confidentiality probabilities and then combining them through a weighted sum. Specifically, define
\begin{align*}
\bar{A}_\omega = &(1-\omega) (\pi_{\{X_t,\hat{X}_t\}}(0,0)+\pi_{\{X_t,\hat{X}_t\}}(1,1))\\
&+ \omega (\pi_{\{X_t,\hat{X}^e_t\}}(0,1)+\pi_{\{X_t,\hat{X}^e_t\}}(1,0)),
\end{align*}
where $\omega\in[0,1]$. Here, $\pi_{\{X_t,\hat{X}_t\}}(0,0)+\pi_{\{X_t,\hat{X}_t\}}(1,1)$ is the marginal reconstruction accuracy at Bob, while $\pi_{\{X_t,\hat{X}^e_t\}}(0,1)+\pi_{\{X_t,\hat{X}^e_t\}}(1,0)$ is the marginal confidentiality term at Eve. By Theorem~\ref{thm:stationary} or equivalently from \cite{SalKouPap:J24}, these terms are given by
\begin{equation*}
\pi_{\{X_t,\hat{X}_t\}}(0,0)
=
\frac{q\bigl(q+p_\alpha p_s(1-q)\bigr)}
{(p+q)\bigl(p+q+p_\alpha p_s(1-p-q)\bigr)},
\end{equation*}
\begin{equation*}
\pi_{\{X_t,\hat{X}_t\}}(1,1)
=
\frac{p\bigl(p+p_\alpha p_s(1-p)\bigr)}
{(p+q)\bigl(p+q+p_\alpha p_s(1-p-q)\bigr)},
\end{equation*}
\begin{equation*}
\pi_{\{X_t,\hat{X}^e_t\}}(0,1)
=
\frac{pq(1-p_\alpha p_s^e)}
{(p+q)\bigl(p+q+p_\alpha p_s^e(1-p-q)\bigr)},
\end{equation*}
\begin{equation*}
\pi_{\{X_t,\hat{X}^e_t\}}(1,0)
=
\frac{pq(1-p_\alpha p_s^e)}
{(p+q)\bigl(p+q+p_\alpha p_s^e(1-p-q)\bigr)}.
\end{equation*}
In particular, $\bar{A}_0$ reduces to the pure accuracy metric, while $\bar{A}_1$ reduces to the pure confidentiality metric.
\begin{figure*}[t]
\begin{centering}
\includegraphics[width=1\textwidth]{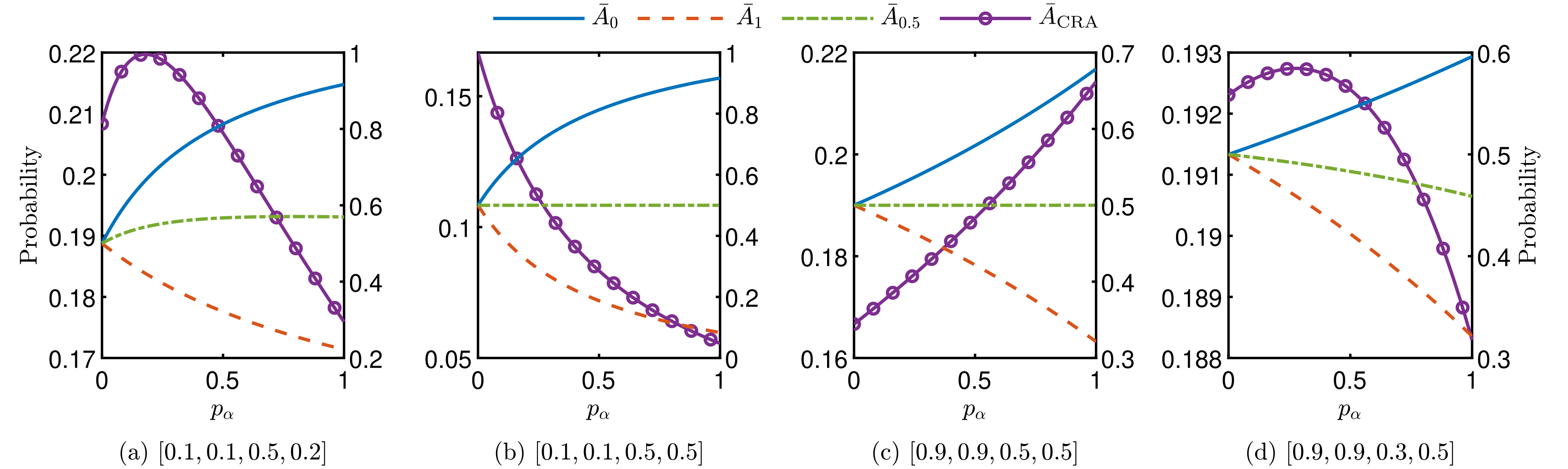}
\par\end{centering}
\caption{\label{fig:P_over_palpha}Probabilities of simultaneous confidentiality and accuracy $\bar{A}_\text{CRA}$, confidentiality only $\bar{A}_1$, accuracy only $\bar{A}_0$, and the balance metric $\bar{A}_{0.5}$ versus the transmission probability under different settings. The metric $\bar{A}_\text{CRA}$ is plotted against the left y-axis, whereas all the other metrics share the right y-axis. The title of each subfigure indicates the corresponding parameter setting $[p, q, p_s, p_s^e]$.}
\end{figure*}

\begin{figure*}[t]
\begin{centering}
\includegraphics[width=1\textwidth]{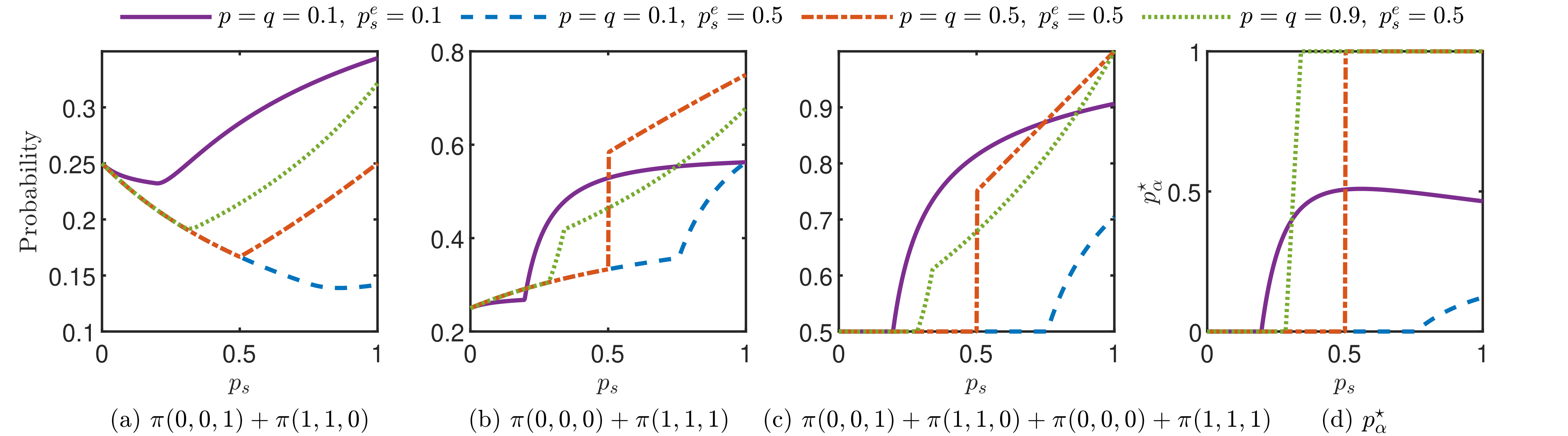}
\par\end{centering}
\caption{\label{fig:P_over_ps} (a) Probability of simultaneous confidentiality and accuracy $\bar{A}_\text{CRA}=\pi(0,0,1)+\pi(1,1,0)$, (b) probability of simultaneous non-confidentiality and accuracy $\pi(0,0,0)+\pi(1,1,1)$, (c) probability of accuracy $\bar{A}_0=\pi(0,0,1)+\pi(1,1,0)+\pi(0,0,0)+\pi(1,1,1)$ under the optimal \ac{rs} policies, and (d) the corresponding optimal transmission probability $p_\alpha^*$.}
\end{figure*}

\subsubsection{Analytical Validation}
Fig.~\ref{fig:Monto_carlo} shows the average simultaneous confidentiality and accuracy $\bar{A}_\text{CRA}$ versus the sampling probability $p_\alpha$. The closed-form expression from Theorem~\ref{thm:stationary}, the numerical stationary solution obtained from the normalized eigenvector associated with the unit eigenvalue of the transition matrix, and the Monte Carlo results are in near-perfect agreement, thereby validating the joint system-dynamics analysis and stationary characterization. In addition, the optimizer and optimal value predicted by Theorem~\ref{thm:opt_P1} coincide with the peak of the curve, confirming the policy characterization.
\begin{figure*}[t]
\begin{centering}
\includegraphics[width=1\textwidth]{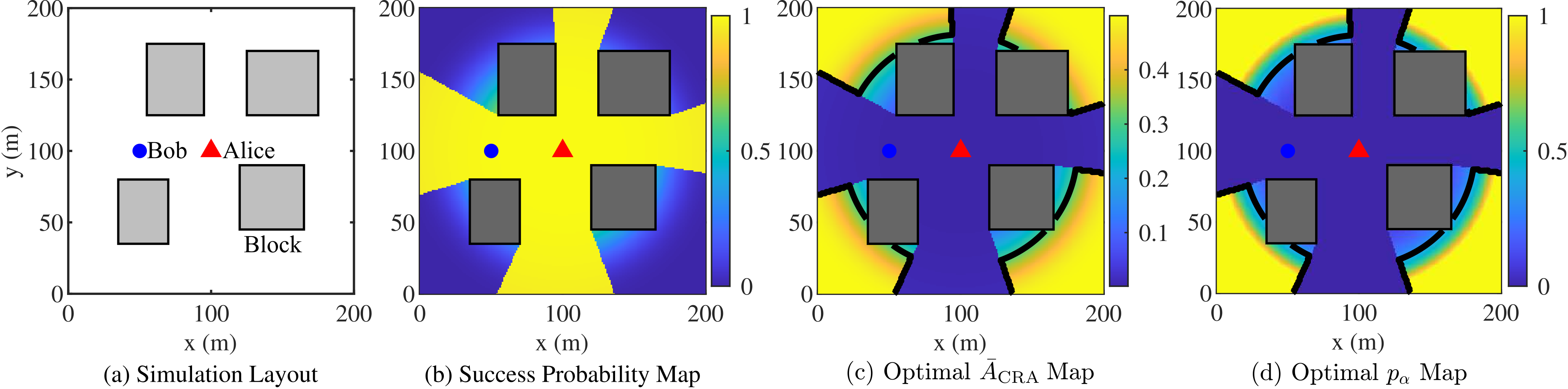}
\par\end{centering}
\caption{\label{fig:geofencing} Simulation layout and spatial maps under arbitrary eavesdropper locations. (a) shows the simulation layout; (b) shows Eve's channel success-probability map; (c) shows the corresponding optimal $\bar{A}_{\text{CRA}}$ map; and (d) shows the corresponding optimal transmission-probability $p_\alpha$ map. The black contour in (c) and (d) represents the geofencing boundary induced by the threshold $\bar{A}_{\text{CRA}}=0.3$.}
\end{figure*}
\subsubsection{\ac{cra}-Induced Insights}
Fig.~\ref{fig:P_over_palpha} shows that the optimal transmission probability for maximizing $\bar{A}_{\text{CRA}}$ is environment-dependent, unlike the marginal metrics $\bar{A}_0$ and $\bar{A}_1$, whose optima are always attained at $p_\alpha=1$ and $p_\alpha=0$, respectively. In Fig.~\ref{fig:P_over_palpha}a, although the legitimate channel is better than the eavesdropping channel, transmitting more frequently is not always beneficial: a larger $p_\alpha$ improves accuracy but degrades confidentiality, and excessive transmissions eventually reduce the probability of being simultaneously confidential and accurate. In Fig.~\ref{fig:P_over_palpha}d, by contrast, even when the eavesdropping channel is better than the legitimate channel, increasing $p_\alpha$ to a proper level can still improve $\bar{A}_{\text{CRA}}$, because the accuracy gain outweighs the confidentiality loss.

Fig.~\ref{fig:P_over_palpha} also shows that secure reconstruction cannot be reduced to a weighted combination of two separate marginal metrics. In Fig.~\ref{fig:P_over_palpha}a, the balance metric $\bar{A}_{0.5}$ both misidentifies the optimal transmission probability and overestimates the achievable performance: it suggests an apparent optimum of about $0.57$ with $p_\alpha = 0.8$, whereas the actual probability of being simultaneously confidential and accurate under the corresponding policy is only about $0.189$, below the true optimum of approximately $0.22$ with less transmission probability $p_\alpha = 0.18$. Figs.~\ref{fig:P_over_palpha}b and \ref{fig:P_over_palpha}c further illustrate an extreme case with $p_s=p_s^e$, where $\bar{A}_{0.5}$ is insensitive to $p_\alpha$, yet the optimal $p_\alpha$ for $\bar{A}_{\text{CRA}}$ still varies significantly with the source dynamics, similar to the extreme analysis in Remark~\ref{rem:extrem_analysis_opt}: when the source evolves slowly, a small $p_\alpha$ is preferred since the state can often be inferred from temporal correlation, whereas when the source evolves rapidly, a large $p_\alpha$ becomes preferable because guessing is no longer reliable.

Fig.~\ref{fig:P_over_ps} reveals three key observations. First, improving the legitimate channel quality $p_s$ does not necessarily increase the simultaneous confidential-and-accurate probability $\pi(0,0,1)+\pi(1,1,0)$ monotonically. As shown in Figs.~\ref{fig:P_over_ps}a--\ref{fig:P_over_ps}c, although the overall probability of accurate reconstruction always increases with $p_s$, the increment may initially be dominated by the non-confidential component $\pi(0,0,0)+\pi(1,1,1)$, which causes the confidential-and-accurate probability to decrease before increasing. Second, when the eavesdropping channel is strong, improving $p_s$ alone may be insufficient to deliver meaningful confidentiality gain. In this case, increasing the source dynamics ($p=q=0.1$ to $p=q=0.9$) or weakening the eavesdropping channel ($p_s^e=0.5$ to $p_s^e=0.1$) is substantially more effective. Third, the optimal transmission probability $p_\alpha^*$ is highly environment-dependent and can even be non-monotonic in $p_s$, as shown in Fig.~\ref{fig:P_over_ps}d. When the source evolves slowly, and the eavesdropping channel is strong, conservative transmission is preferred; when the source evolves faster, aggressive transmission becomes optimal; and when the legitimate channel is already sufficiently reliable, further increasing the transmission frequency may be undesirable due to the additional leakage it incurs.

\subsubsection{\ac{cra}-Induced Geofencing}
To further evaluate the proposed framework in a representative practical scenario, we consider the urban layout shown in Fig.~\ref{fig:geofencing}a, where the transmitter (Alice) is located at the center, and the legitimate receiver (Bob) is placed at a fixed location, while the eavesdropper (Eve) is assumed to be able to appear anywhere in the area. The large-scale propagation is modeled by the 3GPP Urban Micro (UMi) path-loss model in~\cite{TR36814}, where the \ac{los} and \ac{nlos} links are given by
\[
\mathrm{PL}_j=
\begin{cases}
22.0+28.0\log_{10}(d_j)+20\log_{10}(f_c), & \text{LOS},\\
22.7+36.7\log_{10}(d_j)+26\log_{10}(f_c), & \text{NLOS},
\end{cases}
\]
and the small-scale fading is modeled as Rayleigh fading. Based on this channel model, the channel success probability is defined as the probability that the instantaneous channel quality exceeds a prescribed threshold, {\em i.e.}, $\Pr(\gamma>\gamma_{\mathrm{th}})$, which yields the spatial success-probability map shown in Fig.~\ref{fig:geofencing}b. Then, by applying Theorem~\ref{thm:opt_P1} pointwise over space, we obtain the corresponding optimal transmission probability $p_\alpha^*$ and the optimal simultaneous confidentiality-and-accuracy metric $\bar{A}_{\text{CRA}}^*$, as shown in Figs.~\ref{fig:geofencing}c and~\ref{fig:geofencing}d, respectively.

Based on the resulting $\bar{A}_{\text{CRA}}^*$ map, we further define a geofencing region by imposing the constraint $\bar{A}_{\text{CRA}}^*<\tau$, where $\tau$ is a prescribed confidentiality and accuracy threshold. The resulting geofencing boundary is shown by the black contour in Figs.~\ref{fig:geofencing}c and~\ref{fig:geofencing}d. This contour identifies the set of Eve locations that would make the system insufficiently confidential, in the sense that once Eve enters this region, the achievable simultaneous confidentiality-and-accuracy performance falls below the required threshold. It is also observed that the geofencing region becomes relatively large in the left and right directions, where the lack of blockage leads to more favorable propagation for Eve. This observation motivates further study on how to shrink the geofencing region, for example by introducing additional blockages or by designing proactive interference or jamming strategies.

\section{Conclusion}
This paper studied the characterization and optimization of joint accuracy and confidentiality in semantic-aware secure remote reconstruction. By introducing confidential reconstruction accuracy (CRA) and developing a three-dimensional Markov-chain analysis, we derived closed-form expressions for the long-term average CRA and the optimal transmission probability. Numerical results show that the marginal-based analysis fails to capture the intrinsic coupling between legitimate reconstruction and adversarial exposure, and may therefore misidentify the optimal \ac{rs} policy and misestimate the achievable simultaneous accuracy-confidentiality performance. We further revealed nontrivial structural behaviors: more frequent transmissions or better legitimate channels do not necessarily improve joint accurate and confidential reconstruction, and under strong eavesdropping channels, improving the legitimate channel alone may be insufficient, necessitating adjustments to source dynamics or adversary suppression. Finally, we demonstrated the system-level implications of the proposed framework through a geofencing application that establishes spatial safety boundaries for secure remote reconstruction.
\bibliographystyle{IEEEtran}
\bibliography{abrv,ref}
\clearpage
\appendices{}

\section{Proof of Theorem~\ref{thm:stationary}} \label{sec:proof_thm_stationary}
We start by conditioning on the joint age process $(\Theta_t,\Theta^e_t)$. By the law of total probability, we have
\begin{align}
    & \pi(x,a,b) \nonumber\\
    & = \lim_{t\to\infty} \sum_{i=0}^\infty \sum_{j=0}^\infty \Pr(X_t=x, \hat{X}_t=a, \hat{X}^e_t=b \mid \Theta_t=i, \Theta^e_t=j) \nonumber \\
    & \quad \times \Pr(\Theta_t=i,\Theta_t^e=j). \label{eq:stationary_c_age}
\end{align}
Under the receiver estimator model induced by the age variables, as defined in (\ref{eq:def_estimator}), the first term in \eqref{eq:stationary_c_age} becomes
\begin{align*}
    &\Pr(X_t=x, \hat{X}_t=a, \hat{X}^e_t=b \mid \Theta_t=i, \Theta^e_t=j)\\
    & \quad = \Pr(X_t=x, X_{t-i}=a, X_{t-j}=b \mid \Theta_t=i, \Theta^e_t=j)\\
    & \quad = \Pr(X_t=x, X_{t-i}=a, X_{t-j}=b),
\end{align*}
where the last equality holds because the age process is fully determined by $(\alpha_t,H_t)$, as defined in \eqref{eq:def_aoi}, and is independent of the source process $\{X_t\}$ under the \ac{rs} policy.

Since $\{X_t\}$ is a finite-state ergodic Markov chain, and the joint age process $\{(\Theta_t,\Theta_t^e)\}$ is a positive-recurrent Markov chain whose states communicate and state $(0,0)$ has finite expected return time due to $\lambda_{11}>0$, the following limits exist.
\begin{equation*}
    \phi(x,a,b,i,j) := \lim_{t\to\infty} \Pr(X_t=x, X_{t-i}=a, X_{t-j}=b).
\end{equation*}
\begin{equation*}
    \rho(i,j) := \lim_{t\to\infty} \Pr(\Theta_t=i,\Theta_t^e=j),
\end{equation*}
Thus, the joint stationary probability becomes
\begin{equation}
    \pi(x,a,b) = \sum_{i=0}^\infty \sum_{j=0}^\infty \phi(x,a,b,i,j)\rho(i,j). \label{eq:stationary_d_limit}
\end{equation}

Next, we first derive the limiting distributions $\phi(x,a,b,i,j)$ and $\rho(i,j)$ separately, and then combine them to obtain the joint stationary distribution $\pi(x,a,b)$.

\subsection{Stationary Distribution of Joint Age Process}
The stationary distribution of joint age process $\rho(i,j)$ is characterized by three cases.

\textit{Case 1: $i = j$ (Synchronized updates).}
This case corresponds to the event that the most recent successful update is received simultaneously by both Bob and Eve. The age pair is reset to $(0,0)$ with probability $\lambda_{11}$, and then increases synchronously only if neither receiver successfully receives a new update in the subsequent slots. Since a joint failure occurs with probability, $\lambda_{00}$ in each slot, we obtain
\begin{equation}
    \rho(i,i)  = \lambda_{11} (\lambda_{00})^i,\,i\ge0. \label{eq:stationary_a_1}
\end{equation}

\textit{Case 2: $i < j$ (Bob is fresher).}
Let $\Delta = j - i > 0$. Then, at time $t-i$, the system must be in a boundary state $(0, \Delta)$, meaning that Bob has just received a successful update while Eve’s age exceeds Bob’s by $\Delta$. From that time to time $t$, the age pair can evolve from $(0, \Delta)$ to $(i,j)$ only if $i$ consecutive joint failures occur. Therefore, 
\begin{equation}
    \rho(i,j) = \rho(0, \Delta) (\lambda_{00})^i.
    \label{eq:mid_age_pi}
\end{equation}
It remains to determine the boundary term $\rho(0, \Delta)$. The event $(\Theta_{t-i},\Theta_{t-i}^e)=(0, \Delta)$ occurs when Bob successfully receives an update while Eve does not, and Eve’s age in the previous slot equals $\Delta-1$. The former occurs with probability $\lambda_{10}$. Moreover, Eve successfully receives an update in each slot with marginal probability $P_B:=\lambda_{11}+\lambda_{01}$, so its marginal age distribution is geometric 
\begin{equation*}
    \Pr(\Theta^e = k) = P_B (1-P_B)^k,\, k\ge 0
\end{equation*}
Hence,
\begin{equation*}
    \rho(0, \Delta) = \lambda_{10}  P_B (1 - P_B)^{\Delta - 1} .
\end{equation*}
Substituting $\rho(0, \Delta)$ into \eqref{eq:mid_age_pi}, we have 
\begin{equation}
    \rho(i,j) = \lambda_{10} P_B (1 - P_B)^{j - i - 1} (\lambda_{00})^i .\label{eq:stationary_a_2}
\end{equation}

\textit{Case 3: $i > j$ (Eve is fresher).}
This case is symmetric to Case 2, with Bob and Eve interchanged. Hence,
\begin{equation}
    \rho(i,j) = \lambda_{01} P_A (1 - P_A)^{i - j- 1} (\lambda_{00})^j, \label{eq:stationary_a_3}
\end{equation}
where $P_A := \lambda_{11} + \lambda_{10}$ denotes the marginal success probability at Bob.

\subsection{Stationary Three-Time Joint Distribution:}
The stationary three-time joint distribution $\phi(x,a,b,i,j)$ is also characterized by three cases.

\textit{Case 1: $i = j$ (Synchronized updates).}
In this case, the time instants $t-i$ and $t-j$ coincide. Hence, the event has nonzero probability only if $a=b$. Under stationarity, the source is in state $a$ with probability $v_a$, and then evolves from $a$ to $x$ over $i$ steps according to $Q^i$. Therefore,
\begin{equation}
    \phi(x,a,b,i,i) = \delta(a,b) v_a [Q^i]_{a,x}. \label{eq:stationary_3_1}
\end{equation}

\textit{Case 2: $i < j$ (Bob is fresher).}
Here, $t-j<t-i<t$. The process follows the path $\xrightarrow{\infty}b \xrightarrow{j-i} a \xrightarrow{i} x$. By the Markov property and stationarity, we have
\begin{equation}
    \phi(x,a,b,i,j) = v_b [Q^{j-i}]_{b,a} [Q^i]_{a,x}. \label{eq:stationary_3_2}
\end{equation}

\textit{Case 3: $i > j$ (Eve is fresher).}
Here, $t-i<t-j<t$. The process follows the path $\xrightarrow{\infty}a \xrightarrow{i-j} b \xrightarrow{j} x$. Similarly, we have
\begin{equation}
    \phi(x,a,b,i,j) = v_a [Q^{i-j}]_{a,b} [Q^j]_{b,x}. \label{eq:stationary_3_3}
\end{equation}

\subsection{Stationary Distribution of Confidentiality Reconstruction Process}
Having characterized the stationary age distribution $\rho(i,j)$ and the stationary three-time joint distribution, we now derive the stationary distribution $\pi(x,a,b)$. As above, the derivation is carried out by distinguishing three cases.

\textit{Case 1: $i = j$ (Synchronized updates).}
In this case, the two receivers share the same freshest update, and therefore $a=b$. Substituting \eqref{eq:stationary_a_1} and \eqref{eq:stationary_3_1} into \eqref{eq:stationary_d_limit}, we have
\begin{align}
    T_1(x,a,b) &:= \sum_{i=0}^\infty \delta(a,b) v_a [Q^i]_{a,x} \lambda_{11} (\lambda_{00})^i \nonumber \\
    &= \delta(a,b) \lambda_{11} v_a \left[ \sum_{i=0}^\infty (\lambda_{00}Q)^i \right]_{a,x} \nonumber \\
    &= \delta(a,b) \lambda_{11} v_a \left[ (I - \lambda_{00}Q)^{-1} \right]_{a,x}.
\end{align}

\textit{Case 2: $i < j$ (Bob is fresher).}
Let $\Delta = j - i \ge 1$. Substituting \eqref{eq:stationary_a_2} and \eqref{eq:stationary_3_2} into \eqref{eq:stationary_d_limit}, we have
\begin{align}
    & T_2(x,a,b)\nonumber\\ 
    &:= \sum_{i=0}^\infty \sum_{\Delta=1}^\infty v_b [Q^\Delta]_{b,a} [Q^i]_{a,x}  \lambda_{10} P_B (1 - P_B)^{\Delta - 1} (\lambda_{00})^i \nonumber \\
    &= v_b \lambda_{10} P_B \left[ \sum_{\Delta=1}^\infty (1-P_B)^{\Delta-1} Q^\Delta \right]_{b,a} \left[ \sum_{i=0}^\infty (\lambda_{00}Q)^i \right]_{a,x} \nonumber \\
    & = v_b \lambda_{10} P_B \left[ Q(I - (1-P_B)Q)^{-1} \right]_{b,a} \left[ (I - \lambda_{00}Q)^{-1} \right]_{a,x}.
\end{align}

\textit{Case 3: $i > j$ (Eve is fresher).}
By symmetry, let $\Delta = i - j \ge 1$. Substituting \eqref{eq:stationary_a_3} and \eqref{eq:stationary_3_3} into \eqref{eq:stationary_d_limit}, we have
\begin{align}
    T_3(x,a,b) &= v_a \lambda_{01} P_A \left[ Q(I - (1-P_A)Q)^{-1} \right]_{a,b} \nonumber \\
    &\quad \times \left[ (I - \lambda_{00}Q)^{-1} \right]_{b,x}.
\end{align}

Combining the above three cases, we obtain
\begin{equation}
    \pi(x,a,b) = T_1(x,a,b) + T_2(x,a,b) + T_3(x,a,b).
\end{equation}

\section{Proof of Theorem~\ref{thm:opt_P1}} \label{sec:proof_thm_opt_P1}
Before proving Theorem~\ref{thm:opt_P1}, we first characterize the signs of the coefficients. It is straightforward to verify that \(B\le 0\), \(C\le 0\), and \(E\le 0\) for all \(p,q,p_s,p_s^e\in[0,1]\). In contrast, the signs of \(A\)  and \(D\) are parameter-dependent and cannot be determined a priori. Next, to determine the optimal $p_{\alpha}$, we analyze the shape of \(\bar{A}_{\text{CRA}}(p_\alpha)\) via its first- and second-order derivatives.

From Corollary~\ref{cor:average_cra}, the first derivative of \(\bar{A}_{\text{CRA}}(p_\alpha)\) is
\begin{equation*}
\frac{d\bar{A}_{\text{CRA}}}{dp_\alpha}
=
\frac{-AC\,p_\alpha^2-2BC\,p_\alpha+(AE-BD)}
{\left(Cp_\alpha^2+Dp_\alpha+E\right)^2}.
\end{equation*}
Define
\begin{equation*}
M(p_\alpha):=-AC\,p_\alpha^2-2BC\,p_\alpha+(AE-BD).
\end{equation*}
Since $\left(Cp_\alpha^2+Dp_\alpha+E\right)^2>0$ over the feasible interval, the sign of $\bar{A}_{\text{CRA}}'(p_\alpha)$ is completely determined by $M(p_\alpha)$.

To determine the sign of $M(p_\alpha)$, we first analyze the zero points and the vertex of \(M(p_\alpha)\).

The zero points \(M(p_\alpha)=0\) are given by
\begin{equation*}
r_{\pm}=\frac{-BC\pm\sqrt{\Delta}}{AC},
\qquad
\Delta:=B^2C^2+AC(AE-BD).
\end{equation*}
Substituting the explicit expressions of \(A\), \(B\), \(C\), \(D\), and \(E\) into \(\Delta\), we obtain
\begin{equation*}
\Delta
=
p^2q^2\,p_s p_s^e\,(p+q)^4\,(p+q-1)^2\,(p_s-p_s^e)^2\,
(p_s+p_s^e-p_sp_s^e)^2
\end{equation*}
\begin{equation*}
\qquad\times
\Big[p_s+(p+q)p_s^e(1-p_s)\Big]
\Big[p_s^e+p_s(2-p-q)(1-p_s^e)\Big].
\end{equation*}
Since \(p,q,p_s,p_s^e\in[0,1]\) and \(p+q\in[0,2]\), every factor in the above expression is nonnegative. Hence, \(\Delta\ge 0\).

Next, differentiating \(M(p_\alpha)\) yields
\begin{equation*}
\frac{dM}{dp_\alpha}=-2AC\,p_\alpha-2BC=-2C(Ap_\alpha+B),
\end{equation*}
and thus the only vertex of \(M\) is located at
\begin{equation*}
p_\alpha^{\mathrm v}=-\frac{B}{A}.
\end{equation*}

We next distinguish two cases based on the sign of $p_\alpha^{\mathrm v}$.

\textit{Case 1: \(A<0\).}
Since \(B\le 0\), the vertex of \(M(p_\alpha)\), \(p_\alpha^{\mathrm v}=-B/A\le 0\).
Moreover, \(C\le 0\) and \(A<0\) imply \(-AC<0\). Hence, \(M(p_\alpha)\) is concave and therefore decreasing on \([0,1]\).

Since \(AC\ge0\), the two roots satisfy
\begin{equation*}
r_\mathrm{s}=\frac{-BC-\sqrt{\Delta}}{AC}
<
\frac{-BC+\sqrt{\Delta}}{AC}=r_\mathrm{l}.
\end{equation*}
Because the vertex lies to the left of $[0,1]$, the smaller root satisfies \(r_\mathrm{s}<0\) and only the larger root \(r_\mathrm{l}\) may lie in \([0,1]\). Accordingly, \(\bar{A}_{\text{CRA}}(p_\alpha)\) is increasing, decreasing, or increasing-then-decreasing on \([0,1]\) depending on whether \(r_\mathrm{l}\ge 1\), \(r_\mathrm{l}\le 0\), or \(r_\mathrm{l}\in(0,1)\), respectively. Therefore, the optimal solution is obtained by projecting \(r_\mathrm{l}\) onto \([0,1]\), namely,
\begin{equation*}
p_\alpha^*=\left[r_\mathrm{l}\right]_0^1
=
\left[\frac{-BC+\sqrt{\Delta}}{AC}\right]_0^1.
\end{equation*}

\textit{Case 2: \(A>0\).}
Using the explicit expressions of \(A\) and \(B\), we obtain
\begin{equation*}
A+B
=
-pq(1-p_s^e)\Big[(p+q)(1-p_s)(p_s+p_s^e)+2p_s^2\Big]\le 0.
\end{equation*}
Hence,\(B\le -A\). Since \(A>0\), it follows that the vertex of \(M(p_\alpha)\), \(
p_\alpha^{\mathrm v}=-{B}/{A}\ge 1\). Moreover, \(C\le 0\) and \(A>0\) imply \(-AC>0\). Hence, \(M(p_\alpha)\) is convex and therefore decreasing on \([0,1]\).

Since \(AC<0\), the roots satisfy
\begin{equation*}
r_{\mathrm{s}} = \frac{-BC+\sqrt{\Delta}}{AC}
<
\frac{-BC-\sqrt{\Delta}}{AC} = r_{\mathrm{l}}.
\end{equation*}
Because the vertex lies to the right of \([0,1]\), the larger root satisfies \(
r_{\mathrm{l}}>1\) and only the smaller root \(r_{\mathrm{s}}\) may lie in \([0,1]\). Accordingly, \(\bar{A}_{\text{CRA}}(p_\alpha)\) is increasing, decreasing, or increasing-then-decreasing on \([0,1]\), depending on whether \(r_{\mathrm{s}}\ge 1\), \(r_{\mathrm{s}}\le 0\), or \(r_{\mathrm{s}}\in(0,1)\), respectively. Therefore,
\begin{equation*}
p_\alpha^*=\left[r_{\mathrm{s}}\right]_0^1
=
\left[\frac{-BC+\sqrt{\Delta}}{AC}\right]_0^1.
\end{equation*}

The stated result then follows by combining the two cases.

\end{document}